\documentclass{article}

\usepackage{amsmath}
\usepackage{amssymb}
\usepackage{amsfonts}
\usepackage{graphicx}
\newtheorem{theorem}{Theorem}[section]
\newtheorem{lemma}[theorem]{Lemma}
\newenvironment{proof}[1][Proof]{\begin{trivlist}
\item[\hskip \labelsep {\bfseries #1}]}{\end{trivlist}}

\newcommand{\ccNP}{\textrm{\textsc{NP}}}
\newcommand{\qed}{\hfill \ensuremath{\Box}}

\author{Andrew Winslow\footnote{Department of Computer Science, Tufts University. Research supported in part by NSF grants CCF-0830734 and CBET-0941538.}\\
\texttt{awinslow@cs.tufts.edu}}
\title{Inapproximability of the Smallest Superpolyomino Problem}
\date{}

\begin{document}

\maketitle

\begin{abstract}
We consider the \emph{smallest superpolyomino problem}: given a set of colored polyominoes, find the smallest polyomino containing each input polyomino as a subshape.
This problem is shown to be \ccNP-hard, even when restricted to a set of polyominoes using a single common color.
Moreover, for sets of polyominoes using two or more colors, the problem is shown to be \ccNP-hard to approximate within a $O(n^{1/3-\varepsilon})$-factor for any $\varepsilon > 0$. 
\end{abstract}

\section{Introduction}

The smallest superpolyomino problem restricted to one dimension yields the well-known \emph{smallest supersting problem}: given a set of strings (i.e. $n \times 1$ polyominoes), find the smallest string containing each string in the set as a substring. 
This problem is known to be \ccNP-hard and admits a simple greedy 3-approximation~\cite{Blum-1994}.
The approximation algorithm leads to a straightfoward $O(\log^3{n})$-approximation algorithm~\cite{Charikar-2005} for the minimum-size context-free grammar encoding a string, known as the \emph{smallest grammar problem}.

In attempting to extend this idea to solve a higher-dimensional version of the smallest grammar problem, we consider the problem of finding the smallest superpolyomino for a given set of input polyominoes.
We show that an approximation algorithm for the smallest grammar problem in two dimensions is unlikely to exist by proving that the smallest superpolyomino problem is not approximable within any $O(n^{1/3-\varepsilon})$-factor for any $\varepsilon > 0$.

Moreover, this result remains true when constrained to sets of polyominoes using only two colors.
We also show that the problem remains \ccNP-hard even when constrained to sets of polyominoes using a single color.
In this case, geometry alone is responsible for the problem difficulty, as the one-dimensional version can be solved in linear time with a single pass.

For the inapproximability result, we use a reduction from chromatic number of an arbitrary graph.
Given a graph $G = (V, E)$, the reduction uses polyominoes each representing a vertex in $V$, with two polyominoes able to compactly stack atop each other if and only if the two corresponding vertices are independent in $G$.
Forming a smallest superpolyomino then is equivalent to finding a minimum coloring of the vertices of $V$, i.e. $G$'s chromatic number, shown to be \ccNP-hard to approximate within any $(n^{1-\varepsilon})$-factor by Zuckerman~\cite{Zuckerman-2007}.

For the one-color case, a reduction from set cover is used.
A small polyomino is created for each element of the universe, and a single large polyomino encoding the sets is used as a base for the final superpolyomino.
Each element polyomino fits inside the set polyomino at locations corresponding to sets containing the element, with the exception of a single missing cell.
A small superpolyomino then is obtained by placing the element polyominoes at a small number of locations inside the set polyomino, equivalent to covering the universe with a small number of sets.

\section{Definitions}

A polyomino $P = (S, C)$ is defined by a connected set of points on the square lattice (called \emph{cells}) containing $(0, 0)$, and a coloring $C: S \rightarrow \sigma$ mapping the cells of $P$ to a set of colors $\sigma$, e.g. cell $(3, 1)$ is red, cell $(3, 2)$ is gray, etc.
As shorthand, we denote the color of the cell $C((x, y))$ as $P(x, y)$, and $|P|$ as the number of cells in $P$, i.e. the \emph{size} of $P$.
A \emph{translation} of $P$ is a polyomino $P' = (S', L')$ and coordinate offset $(\delta_x, \delta_y)$ such that $S' = \{(x + \delta_x, y + \delta_y) \mid (x, y) \in S\}$ and $C'((x + \delta_x, y + \delta_y)) = C(x, y)$. 

Two polyominoes $P_u = (S_u, C_u)$ and $P_v = (S_v, C_v)$ at some translation $(\delta_x, \delta_y)$ are \emph{compatible} if for each $(x, y)$, either $P_v(x, y)$ or $P_u(x, y)$ is empty or $P_v(x, y) = P_u(x + \delta_x, y + \delta_y)$.
Similarly, a polyomino $P = (S, C)$ is a \emph{superpolyomino} of $P' = (S', C')$ if there exists a translation $(\delta_x, \delta_y)$ of $P'$ such that for each $(x, y)$, either $(x+\delta_x, y+\delta_y) \not \in S'$ or $P'(x+\delta_x, y+\delta_y) = P(x, y)$, i.e. there is a translation of $P'$ such that $P'$ is compatible with $P$ and lies entirely in $P$. 

\section{Reduction}
\label{sec:reduction}

Given a graph $G = (V, E)$, each vertex $v \in V$ is converted into a polyomino $P_v = (S_v, C_v)$ that encodes $v$ and the neighbors of $v$ in $G$ (see Figure~\ref{fig:reduction-ex}).
Each $P_v$ is a rectangular $2|V| \times |V|$ polyomino with up to $|V|-1$ single squares removed and lower-left corner at $(0, 0)$. 
The four corners of all $P_v$ have a common set of four colors: green, blue, purple, and orange.
Cells at locations $\{(2i + 1, 1) \mid 0 \leq i < |V|\}$ are colored black if $v_i = v$, $r$ if $(v, v_i) \in E$, or are empty locations if $v_i$ is not $v$ or a neighbor of $v$.
All remaining cells have a common gray color.

\begin{figure}[ht]
\centering
\includegraphics[width=0.7\columnwidth]{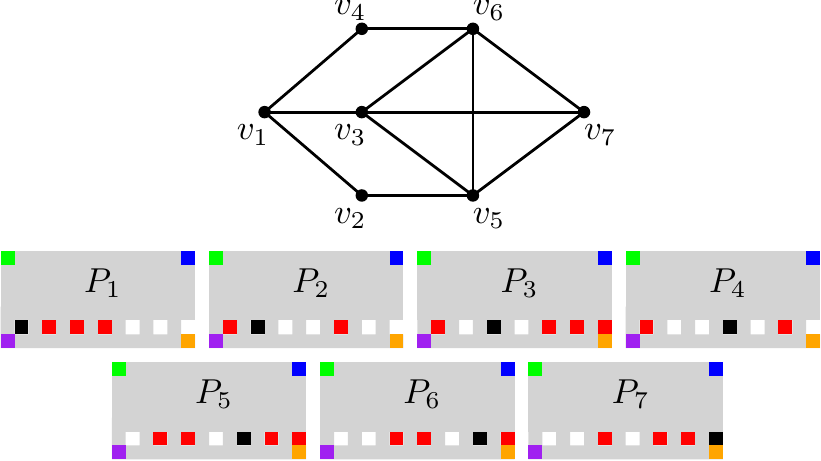}
\caption{An example of the set of polyominoes generated from an input graph by the reduction.}
\label{fig:reduction-ex}
\end{figure}

Consider how two polyominoes $P_u$ and $P_v$ can overlap depending upon the relationship of $u$ and $v$.
Because of the four distinct corner colors, $P_u$ and $P_v$ can only overlap when these four locations in $P_v$ are translated to the same locations in $P_u$.
In this translation, the cells at location $(2i+1, 1)$ in $P_u$ and $P_v$ are compatible exactly when $(u, v) \not\in E$, i.e. $u$ and $v$ are not neighbors. 
All other cells are colored gray and are thus compatible.

The superpolyomino formed by a pair of compatible $P_u$ and $P_v$ in this translation has the common set of four colored corner cells and many gray cells, and has two black cells and a number of red cells corresponding to the combined neighborhoods of $u$ and $v$.
Then by induction, any set of polyominoes can overlap if and only if they form an independent set.
Moreover, if they overlap, they overlap using a set of translations in which the four corners of all polyominoes are placed at four common locations. 

Because the polyominoes can only overlap in this constrained way, any superpolyomino of the polyominoes $\{P_v \mid v \in V\}$ consists of a number of \emph{decks} of superimposed polyominoes corresponding to independent sets of vertices in $G$ arranged disjointly to form a single connected polyomino (see Figure~\ref{fig:reduction-soln-ex}). 

\begin{figure}[ht]
\centering
\includegraphics[width=0.7\columnwidth]{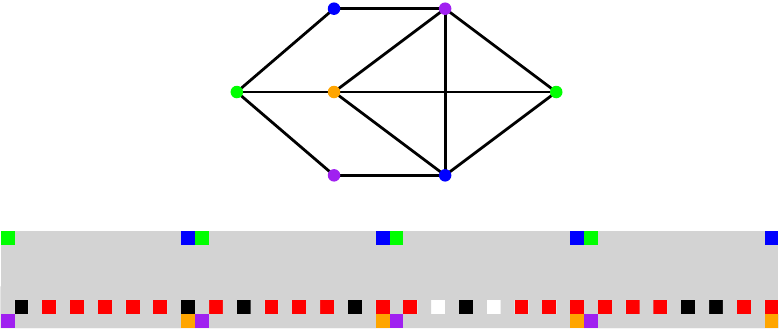}
\caption{An example of a corresponding 4-deck superpolyomino and 4-colored graph.}
\label{fig:reduction-soln-ex}
\end{figure}

Recall that each $P_v$ is a $2|V| \times |V|$ rectangle with $|V|$ cells colored black, red, or are not present.
The size of $P_v$ is then between $2|V|^2 - |V| + 1$ and $2|V|^2$ depending upon the number of neighbors of $v$, and each deck of polyominoes also has size in this range.

\begin{lemma}
\label{lem:deck-to-coloring}
For a graph $G = (V, E)$, there exists a superpolyomino of size at most $2k|V|^2$ for polyominoes $\{P_v \mid v \in V\}$ if and only if the vertices of $V$ can be $k$-colored.
\end{lemma}

\begin{proof}
First, consider extreme sizes of superpolyominoes consisting of $k$ and $k-1$ decks.
For any $V$ and $k$ with $1 \leq k \leq |V|$: 
$$(k-1)(2|V|^2) = 2k|V|^2 - 2|V|^2 < 2k|V|^2 - k|V| = k(2|V|^2 - |V|)$$
That is, the size of any superpolyomino of $k-1$ decks is smaller than the size of any superpolyomino of $k$ decks.

We now prove both implications of the lemma.
First, assume the superpolyomino of size at most $2k|V|^2$ exists.
Then the superpolyomino must consist of at most $k$ decks.
Each deck is the superposition of a set of polyominoes forming an independent set, so $G$ can be $k$-colored.

Next, assume that $G$ can be $k$-colored.
Then the polyominoes $\{P_v \mid v \in V\}$ can be translated to form $k$ decks, one for each color, each with size at most $2|V|^2$.
Placing these decks adjacent to each other yields a superpolyomino of size at most $2k|V|^2$.
\qed
\end{proof}

Note that only $|V|$ cells of each $P_v$ are distinct and depend on $v$, while the other $2|V|^2-|V|$ cells are identical for all $P_v$.
The extra cells are needed for the first inequality in Lemma~\ref{lem:deck-to-coloring}, and they effectively ``drown out'' variation in the size of each deck due to missing cells.

\begin{theorem}
\label{thm:ssp-inapprox}
The smallest superpolyomino problem is \ccNP-hard to approximate within a factor of $O(n^{1/3 - \varepsilon})$ for any $\varepsilon > 0$.
\end{theorem}

\begin{proof}
Consider the smallest superpolyomino problem for the polyominoes generated from a graph $G = (V, E)$ with chromatic number $k$.
There are $|V|$ of these polyominoes, each of size $\Theta(|V|^2)$, so the polyominoes have total size $n = \Theta(|V|^3)$.
By Lemma~\ref{lem:deck-to-coloring}, a superpolyomino of size between $(2|V|^2-|V|)k'$ and $2|V|^2k'$ exists if and only if there exists a $k'$-coloring of $G$.
Then by Zuckerman~\cite{Zuckerman-2007}, finding a superpolyomino such that $(2|V|^2-|V|)k'/(2|V|^2)k = O(|V|^{1-\varepsilon}) = \Theta(n^{1/3-\varepsilon})$ is \ccNP-hard.
\end{proof}

\section{Extensions}

The result of Theorem~\ref{thm:ssp-inapprox} also holds when constrainted to sets of polyominoes using at most two colors by converting each cell into a unique $8 \times 8$ \emph{macrocell} (seen in Figure~\ref{fig:bicolor-adaption}).
The cells on the boundary of the macrocell are colored black and the interior cells are colored gray, except for the three cells $(2, 2)$, $(3, 3)$, and $(4, 4)$. 
These cells are used to encode a binary value between 0 and 7 corresponding to the color of the cell in the original polyomino, with gray indicating a 0 and black a 1.

\begin{figure}[ht]
\centering
\includegraphics[width=1.0\columnwidth]{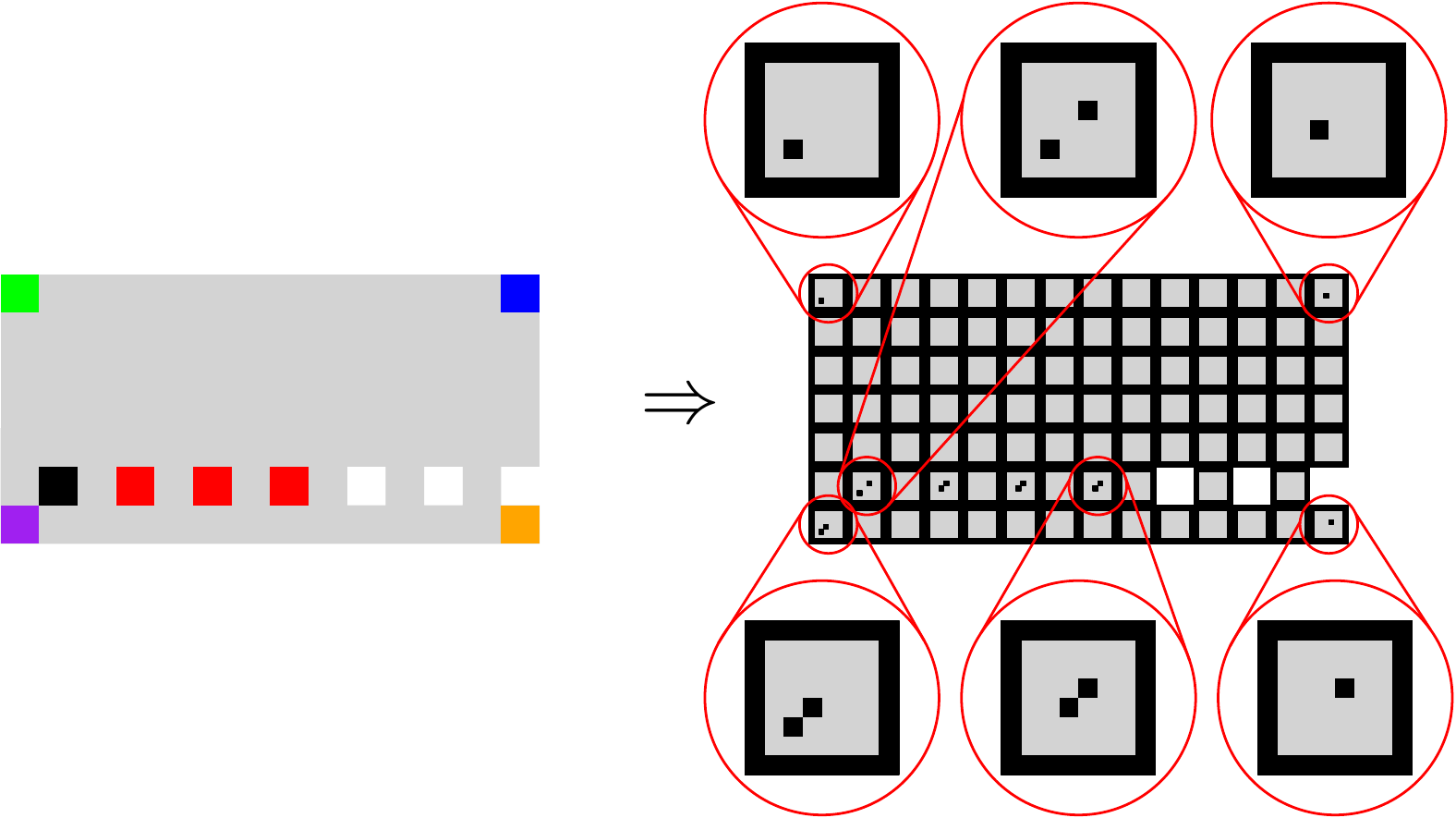}
\caption{Converting a reduction polyomino (left) to a two-color reduction polyomino (right).}
\label{fig:bicolor-adaption}
\end{figure}

Because all polyominoes are simply scaled by a factor of $8 \cdot 8 = 64$, following the proof of Lemma~\ref{lem:deck-to-coloring} gives a similar result, but with a superpolyomino size of $128k|V|^2$ instead of $2k|V|^2$.
Because each polyomino still has size $\Theta(|V|^2)$, the same inapproximability ratio as Theorem~\ref{thm:ssp-inapprox} holds:

\begin{theorem}
The smallest superpolyomino problem for two-color polyomino sets is \ccNP-hard to approximate within a factor of $O(n^{1/3 - \varepsilon})$ for any $\varepsilon > 0$.
\end{theorem}

We now consider the smallest superpolyomino problem for one-color polyomino sets.
A reduction from set cover is done using families of polyominoes as seen in Figures~\ref{fig:monocolor-polyomino} and~\ref{fig:monocolor-reduction}.
Given a universe $U = \{1, 2, \dots, n\}$ and set of sets $S = \{S_1, \dots, S_m\}$ with $\bigcup_{1 \leq i \leq m}S_i = U$, two types of polyominoes are created.
First, a set of \emph{element polyominoes} $\{P_1, P_2, \dots, P_m\}$ is created, each consisting of a $(n+1) \times (n+1)$ \emph{base}, a $1 \times 2n$ \emph{flagpole}, and a $n \times 1$ \emph{flag}.
The lower-left corner of the three components for polyomino $i$ are placed at $(0, 0)$, $(0, n)$, and $(1, n + 2i)$, respectively. 

\begin{figure}[ht]
\centering
\includegraphics[width=0.33\columnwidth]{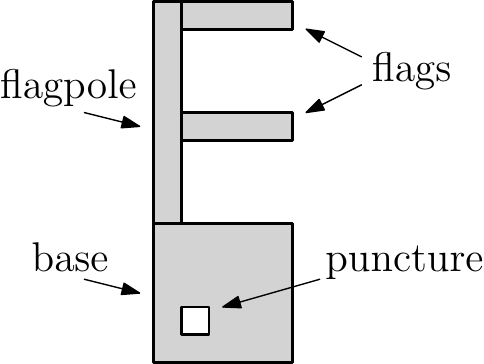}
\caption{The components of element polyominoes and set gadgets used in the one-color reduction.}
\label{fig:monocolor-polyomino}
\end{figure}

Second, a \emph{set polyomino} $\overline{P}$ is created, consisting of a set of \emph{punctured bases} (a base with cell $(1, 1)$ missing) arranged horizontally and attached by single cells, each with a flagpole and some number of flags (called a \emph{set gadget}).
Each flagpole corresponds to a set $S_i = \{e_1, e_2, \dots, e_l\}$, with flags placed at location $((n+2)(i-1) + 1, n + 2e_j)$ for $1 \leq j \leq l$.

\begin{figure}[ht]
\centering
\includegraphics[width=0.9\columnwidth]{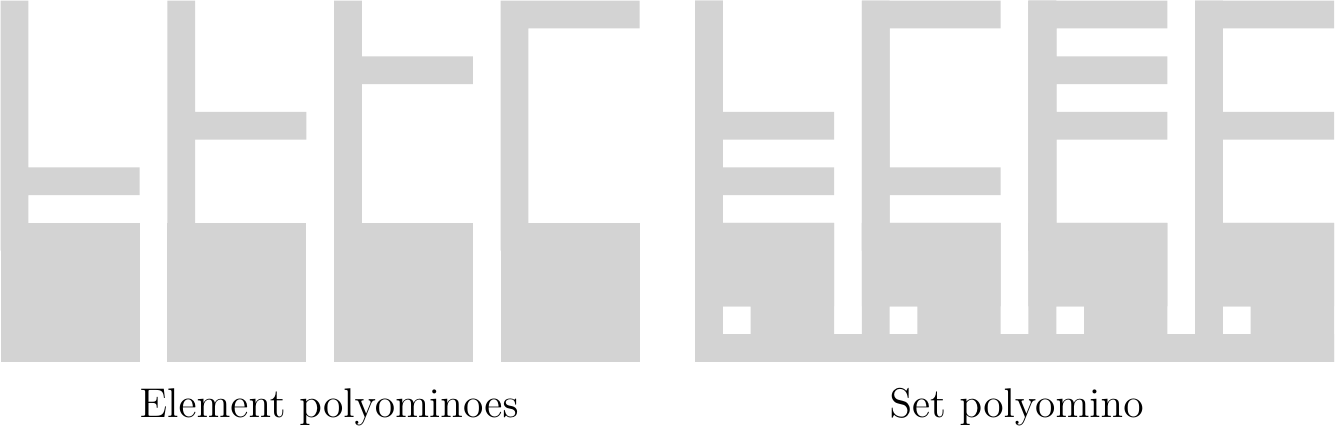}
\caption{The set of polyominoes produced from the reduction from minimum set cover to smallest superpolyomino for the set $\{\{1, 2\}, \{1, 4\}, \{2, 3, 4\}, \{2, 4\}\}$.}
\label{fig:monocolor-reduction}
\end{figure}

Consider forming a superpolyomino of the set polyomino $\overline{P}$ and some element polyomino $P_i$ corresponding to an element $i \in U$.
Define a translation of $P_i$ that places the lower-left corner of its base at the lower-left corner of a punctured base of $\overline{P}$ as an \emph{alignment} of the bases.
Aligning the base of $P_i$ with a punctured base of $\overline{P}$ corresponding to a set containing $i$ results in a superpolyomino of size $|\overline{P}|+1$.
If the base of $P$ is not aligned with a punctured base of $\overline{P}$, then the resulting superpolyomino is at least as large as some superpolymino of $\overline{P}$ and \emph{all} element polyominoes:

\begin{lemma}
\label{lem:no-align-has-big-cost}
Any superpolyomino of the set polyomino $\overline{P}$ and some element polyomino $P_i$ in which the base of $P_i$ is not aligned with the punctured base of a set gadget not corresponding to a set containing has size at least $\overline{P} + n$. 
\end{lemma}

\begin{proof}
First, if the base of $P_i$ is aligned with the punctured base of a set gadget that does not correspond to a set containing $i$, then the cells of the flag of $P_i$ do not overlap with the cells of $\overline{P}$.
The flag has $n$ cells, and so the superpolyomino has at least $|\overline{P}| + n$.

Next, if the base of $P_i$ is \emph{not} aligned with any punctured base in $\overline{P}$, then it must be horizontally or vertically misaligned (or both).
If the base is horizontally misaligned, then some column that does not contain any punctured base of $\overline{P}$ contains cells from the base of $P_i$.
Such a column then must contain at least $n+1$ cells of $P_i$ and at most 1 cell of $\overline{P}$, so the superpolyomino has size at least $|\overline{P}| + n$.

If the base is not horizontally misaligned but is vertically misaligned, then some row containing a row of the base of $P$ has at most 1 cell of $\overline{P}$.
This row is either: 1. a row entirely below or above the set gadget $P$ is aligned with or 2. a row that contains a cell of the flagpole but no flag. 
In either case, the superpolyomino has at least $n+1$ cells of $P$ and at most 1 cell of $\overline{P}$, so the superpolyomino has size at least $|\overline{P}| + n$.
\end{proof}

Intuitively, Lemma~\ref{lem:no-align-has-big-cost} implies that ``playing by the rules'' and aligning the base of each element polyomino with some punctured base of the set polyomino $\overline{P}$ is always better than ``cheating'' by not aligning some element polyomino with any punctured base.
So finding a minimum superpolyomino is equivalent to deciding which punctured base each element polyomino's base should be aligned with, with a unit cost for each punctured base ``patched'' by an element polyomino's aligned base (see Figure~\ref{fig:monocolor-reduction-soln}).

\begin{figure}[ht]
\centering
\includegraphics[width=0.43\columnwidth]{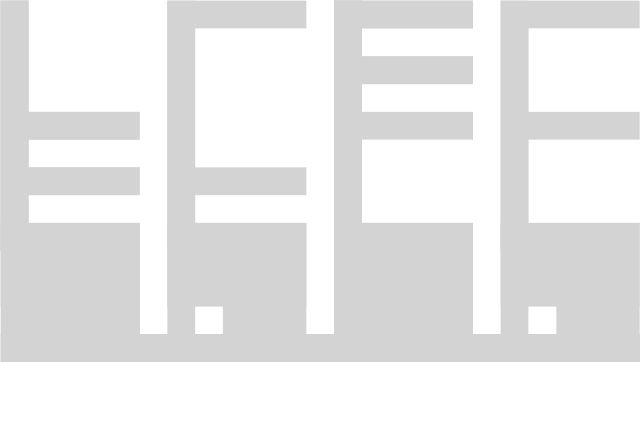}
\caption{The smallest superpolyomino of the polyominoes in Figure~\ref{fig:monocolor-reduction}, corresponding to the set cover $\{S_1, S_3$\}.}
\label{fig:monocolor-reduction-soln}
\end{figure}

\begin{theorem}
The smallest superpolyomino problem for one-color polyomino sets is \ccNP-hard.
\end{theorem}

\begin{proof}
By Lemma~\ref{lem:no-align-has-big-cost}, the smallest superpolyomino of any set of polyominoes generated by the reduction is one resulting from aligning the base of each element polyomino $P_i$ base with the punctured base of a set gadget of $\overline{P}$ corresponding to a set containing $i$.
So a superpolyomino of size $|\overline{P}| + k$ exists if any only if there exists a set cover of size $k$.
Set cover is known and loved by many, and is also \ccNP-hard~\cite{Karp-1972}. 
\end{proof}



\begin{thebibliography}{4}
\bibitem{Blum-1994}
A. Blum, T. Jiang, M. Li, J. Tromp, M. Yannakakis,
Linear approximation of shortest superstrings,
\emph{Journal of the ACM},
41(4)~(1994), 630--647.
\bibitem{Charikar-2005}
M. Charikar, E. Lehman, A. Lehman, D. Liu, R. Panigrahy, M. Prabhakaran, A. Sahai, a. shelat,
The smallest grammar problem,
\emph{IEEE Transactions on Information Theory},
51(7)~(2005), 2554--2576.
\bibitem{Karp-1972}
R. M. Karp, 
Reducibility among combinatorial problems,
\emph{Complexity of Computer Computations},
Plenum, New York, 1972, 85--103.
\bibitem{Zuckerman-2007}
D. Zuckerman,
Linear degree extractors and the inapproximability of max clique and chromatic number,
\emph{Theory of Computation},
3~(2007), 103--128.
\end{thebibliography}
\end{document}